\documentclass[conference]{IEEEtran}

\usepackage{setspace}

\usepackage{graphics,
           psfrag,
           epsfig,
           amsthm,
           cite,
           amssymb,
           url,
           dsfont,
           subfigure,
           algorithm,
           algorithmic,
           balance,
           enumerate,
           color,
           setspace
}
\usepackage{amsmath}%[centertags][tbtags][sumlimits][nosumlimits][intlimits][namelimits][nonamelimits]

\usepackage{epstopdf}

\newtheorem{definition}{Definition}

\newtheorem{theorem}{Theorem}

\newcommand{\eref}[1]{(\ref{#1})}
\newcommand{\sref}[1]{Section~\ref{#1}}

\newcommand{\fref}[1]{Figure~\ref{#1}}

\newcommand{\cref}[1]{Constraint~\ref{#1}}

\newcommand{\ignore}[1]{}

\IEEEoverridecommandlockouts
\begin{document}

\title{Completion Time Reduction in Instantly Decodable Network Coding Through Decoding Delay Control}
\author{
   \authorblockN{Ahmed Douik$^{\dagger}$, Sameh Sorour$^\ast$, Mohamed-Slim Alouini$^\dagger$, and Tareq Y. Al-Naffouri$^{\dagger\ast}$\\}%
   \authorblockA{$^\dagger$King Abdullah University of Science and Technology (KAUST), Kingdom of Saudi Arabia \\
    $^\ast$King Fahd University of Petroleum and Minerals (KFUPM), Kingdom of Saudi Arabia \\
    Email: $^\dagger$\{ahmed.douik,slim.alouini,tareq.alnaffouri\}@kaust.edu.sa \\
    $^\ast$\{samehsorour,naffouri\}@kfupm.edu.sa}
    }

\maketitle

\IEEEoverridecommandlockouts

\begin{abstract}
For several years, the completion time and the decoding delay problems in Instantly Decodable Network Coding (IDNC) were considered separately and were thought to completely act against each other. Recently, some works aimed to balance the effects of these two important IDNC metrics but none of them studied a further optimization of one by controlling the other. In this paper, we study the effect of controlling the decoding delay to reduce the completion time below its currently best known solution. We first derive the decoding-delay-dependent expressions of the users' and their overall completion times. Although using such expressions to find the optimal overall completion time is NP-hard, we use a heuristic that minimizes the probability of increasing the maximum of these decoding-delay-dependent completion time expressions after each transmission through a layered control of their decoding delays. Simulation results show that this new algorithm achieves both a lower mean completion time and mean decoding delay compared to the best known heuristic for completion time reduction. The gap in performance becomes significant for harsh erasure scenarios.
\end{abstract}


\section{Introduction} \label{sec:intro}

\emph{Network Coding (NC)} gained much attention in the past decade after its first introduction in the seminal paper \cite{850663}. In the last lustrum, an important subclass of network coding, namely the Instantly Decodable Network Coding (IDNC) has been an intensive subject of research \cite{6655395,ref4,6725590,ref5,ref7,refsameh,ref6,6766433,ref8,ref17,ref18,arg1,6120247,refahmed} thanks to its several benefits, such as the use of simple binary XOR to encode and decode packets. Moreover, it requires no buffer and allow fast progressive decoding of packets, which is much favorable in many applications (e.g. roadside to vehicle safety messages, satellite networks and IPTV) compared to the long buffering time needed in other NC approaches before decoding.

For as long as the research on IDNC has existed, there were two main metrics that were considered in the literature as measures of its quality, namely the \emph{completion time} \cite{ref4} and the \emph{decoding delay} \cite{ref2}. The former measures how fast the sender can complete the delivery and recovery of requested packets whereas the latter measures how far the sender is from being able to serve all the unsatisfied users in each and every transmission. For several years, these two metrics were considered for optimization separately in many works. Though both were proved to be NP-hard parameters to minimize, many heuristics has been developed to solve them in many scenarios \cite{ref4,ref2,ref5,refsameh,refahmed}, but again separately. In fact, it can be easily inferred from \cite{ref4} and \cite{ref2} that the policies derived so far to optimize one usually degrades the other.

It was not until very recently that one work \cite{nada} has aimed to derive a policy that can balance between these two metric and achieve an intermediate performance for both of them. Nonetheless, to the best of our knowledge, There is no work that aims to explore how these two metrics can be controlled together in order to achieve an even better performance than the currently best known solutions. For instance, every time an unsatisfied user receives a coded packet that is not targeting him, its decoding delay increases and so does its individual completion time. Although this fact was noted for erasure-free transmissions in \cite{nada}, it was used to strike a balance in performance between both metrics and not to investigate whether a smart control of such decoding delay effects will further reduce the overall completion time compared to its current best achievable performance.

In this paper, we aim to design a new completion time reduction algorithm through decoding delay control. We first derive more a general expressions of the individual and overall completion times over erasure channels as a function of the users' decoding delays. Since finding the optimal schedule of coded packets to minimize the overall completion time is NP-hard \cite{arg1}, we use a heuristic that aims to minimize the probability of increasing the maximum of these decoding-delay-dependent completion time expressions after each transmission. This process can be done by partitioning the \emph{IDNC graph} into layers with descending order of user completion time criticality before each transmission. The coding combination for this transmission is then designed by going through these descending order layers sequentially and selecting the combination that minimizes the probability of any decoding delay increments within each layer. This is done while maintaining the instant decodability constraint of the overall coding combination for the targeted users in the more critical layer(s). Finally, we compare through simulations the performance of our designed algorithm to the best known completion time and decoding delay reduction algorithms.

The rest of this paper is organized as follows. \sref{sec:sys} introduces the system model and parameters. In \sref{sec:formulation}, we derive the decoding delay dependent completion time expressions and introduce the problem formulation. The algorithm to solve the problem is illustrated in \sref{sec:algo} and is tested through simulation results in \sref{sec:results}. Finally, we conclude this paper in \sref{sec:conclusion}.

\section{System Model and Parameters} \label{sec:sys}

The model, we consider in this paper, consists of a wireless sender that is required to deliver a frame (denoted by $\mathcal{N}$) of $N$ source packets to a set (denoted by $\mathcal{M}$) of $M$ users. Each user is interested in receiving the $N$ packets of $\mathcal{N}$. In an \emph{initial phase}, The sender transmits the $N$ packets of the frame uncoded. Each user listens to all transmitted packets and feeds back to the sender an acknowledgement for each successfully received packet.

After the \emph{initial phase}, two sets of packets are attributed to each user $i$ at the sender:
\begin{itemize}
\item The \emph{Has} set (denoted by $\mathcal{H}_i$) is defined as the set of packets successfully received by user $i$.
\item The \emph{Wants} set (denoted by $\mathcal{W}_i$) is defined as the set of packets that are lost by user $i$. In other words,we have $\mathcal{W}_i = \mathcal{N} \setminus \mathcal{H}_i$.
\end{itemize}

After the \emph{initial phase}, the \emph{recovery phase} begins. In this phase, the sender exploits the diversity of received packets at the different users to transmit network coded combinations of the source packets. After each transmission, users update the sender in case they receive the coded packet and decode a missing source packets from it. This process is repeated until all users complete the reception of all the packets. Let $p_{i}$, $i\in\mathcal{M}$, be the erasure probability of a packet at user $i$, which is assumed to be constant during the frame period.
For ease of notation, we will assume that the time index $t$ denotes the transmission number within the recovery phase and thus $t=0$ refers to its beginning. In the \emph{recovery phase}, the encoded packets at time $t$ can have one of the following three options for each user $i$:
\begin{itemize}
\item \emph{Non-innovative:} A packet is non-innovative for user $i$ if all the source packets combined in it are from $\mathcal{H}_i(t)$.
\item \emph{Instantly Decodable:} A packet is instantly decodable for user $i$ if it contains \emph{only one source packet} from $\mathcal{W}_i(t)$.
\item \emph{Non-Instantly Decodable:} A packet is non instantly decodable for user $i$ if it contains two or more source packets from $\mathcal{W}_i(t)$.
\end{itemize}
We define the targeted users by a coded packet (or a transmission) as the users for which this packet is instantly decodable. Given a schedule $S$ of coded packets transmitted by the sender, we define the individual completion time, overall completion time and the decoding delay, like in \cite{ref2,nada}, as follows:
\begin{definition}
The individual completion time $\mathcal{C}_i(S)$ of user $i$ is the number of recovery transmissions required until this user obtained all its requested packets.
\end{definition}
\begin{definition}
The overall completion time $\mathcal{C}(S)$ of a frame is the number of recovery transmissions required until all users obtain all their requested packets. It easy to infer that $\mathcal{C}(S) = \max_{i\in\mathcal{M}} \mathcal{C}_i(S)$.
\end{definition}
\begin{definition}
At any recovery phase transmission at time $t$, a user $i$, with non-empty Wants set, experiences a one unit increase of decoding delay if it successfully receives a packet that is either non-innovative or non-instantly decodable. Consequently, the decoding delay $D_i(S)$ experienced by user $i$ given a schedule $S$ is the number of received coded packets by $i$ before its individual completion, which are non-innovative or non-instantly decodable.
\end{definition}

The possible coded combinations for the transmission at time $t$ are determined using the IDNC graph $\mathcal{G}(t)$ \cite{ref2}. This graph is constructed by generating a vertex $v_{ij}$ for every packet $j\in\mathcal{W}_i(t)$ and $\forall~i\in\mathcal{M}$. Two vertices $v_{ij}$ and $v_{kl}$ are adjacent in this graph, and thus can be served simultaneously, if $j=l$ (in which case $i$ and $k$ can be served by simply sending $j$) or $j\in\mathcal{H}_k$ and $l\in\mathcal{H}_i$ (in which case $i$ and $k$ can be served by sending $j\oplus l$). It is easy to infer that this simultaneous service property extends to every clique in the graph \cite{arg1}. In other words, all the users identified by the vertices of a clique $\kappa(t)$ can be simultaneously served by combining the packets identified by the same vertices of that clique $\kappa(t)$. In the rest of the paper, we will designate the transmission occurring at time $t$ by the selected clique $\kappa(t)$ for this transmission from the IDNC graph.

\section{Problem Formulation using Decoding-Delay-Dependent Expressions} \label{sec:formulation}

The following theorem introduces a decoding-delay-dependent expression for the individual completion time of user $i$ and the overall completion time, given the transmission of schedule $S$ from the sender over erasure channels.
\begin{theorem}
For a relatively large number of packets $N$, and a schedule $S$ of transmitted packets by the sender until the overall completion time occurs to all users, the individual completion time for user $i$ can be approximated by:
\begin{equation}\label{eq:ICT}
\mathcal{C}_i(S) \approx \frac{\left|\mathcal{W}_i(0)\right| + D_i(S)-p_i}{1-p_i}
\end{equation}
Consequently, the overall completion time for the same schedule $S$ can be expressed as:
\begin{equation}
\mathcal{C}(S) \approx \max_{i\in\mathcal{M}}\left\{\frac{\left|\mathcal{W}_i(0)\right| + D_i(S)-p_i}{1-p_i}\right\}
\end{equation}
\end{theorem}
\begin{proof}
The proof can be found in Appendix A.
\end{proof}
In the rest of the paper, we will use the approximation with equality as it indeed holds for large $N$.
We can thus formulate the minimum completion time problem as finding the schedule of coded packet $S^*$, such that:
\begin{align}\label{eq:opt}
S^* &= \arg\min_{S\in\mathcal{S}} \left\{\mathcal{C}(S)\right\} \nonumber \\ &= \arg\min_{S\in\mathcal{S}}\left\{ \max_{i\in\mathcal{M}}\left\{ \frac{|\mathcal{W}_i(0)| + D_i(S)-p_i}{1-p_i} \right\}\right\}\;,
\end{align}
where $\mathcal{S}$ is the set of all possible transmission schedules of coded packets.

Clearly, finding this optimal schedule at time $t=0$ through the above optimization formulation is very difficult. This is true due to the dynamic nature of erasures and the dependence of the optimal schedule of their effect, which makes the above equations anti-causal (i.e. current result depends on input from the future). Moreover, we know from the literature that optimizing the completion time over the whole \emph{recovery phase} is intractable\cite{refsameh}, even for the erasure-free scenario \cite{nada}. On the other hand, this formulation shows that the only terms affected by the schedule in the individual and overall completion time expressions are the decoding delay terms of the different users. Consequently, controlling such decoding delays in a smart way throughout the selection of the coded packet schedule can indeed affect the reduction of the completion time significantly. We will thus design a new heuristic algorithm in the next section that takes this fact into consideration.

\section{Design of Heuristic Algorithm}\label{sec:algo}

\subsection{Critical Criterion}
From \eref{eq:opt}, we can see that the optimal schedule is the one that achieves the minimum overall growth in the individual completion time expressions in \eref{eq:ICT}, $\forall~i\in\mathcal{M}$. Since we know that finding such schedule for the entire \textit{recovery phase}, prior to its start, is intractable, we will design our heuristic algorithm such that, in each transmission a time $t>0$, it minimizes the probability of increase of the maximum of such expressions over all users compared to their state before this transmission. To formally express this criterion, let us first define $D_i(t)$ as the individual experienced decoding delay of user $i$ until time $t$. Also, define $\mathcal{C}_i(t)$ as:
 \begin{equation} \label{eq:Ct}
 \mathcal{C}_i(t) = \frac{|\mathcal{W}_i(0)| + D_i(t) - p_i}{1-p_i}
 \end{equation}
In other words, $C_i(t)$ is the anticipated individual completion time of user $i$ if it experiences no further decoding delay increments starting from time $t$. Thus, the philosophy of our proposed heuristic algorithm is to transmit the coded packet $\kappa(t)$ at time $t$ such that:
\begin{equation}\label{eq:heuristic-criterion}
\kappa^{*}(t) = \arg\min_{\kappa(t)\in\mathcal{G}(t)} \left\{\mathds{P}\left(\max_{i\in\mathcal{M}} \left\{\mathcal{C}_i(t)\right\} > \max_{i\in\mathcal{M}} \left\{\mathcal{C}_i(t-1)\right\}\right) \right\}
\end{equation}
We will refer to \eref{eq:heuristic-criterion} as the critical criterion. Let $\mathcal{P}(t)$ be the set of users that can potentially increase $\max_{i\in\mathcal{M}} \left\{\mathcal{C}_i(t)\right\}$ at time $t$ compared to $\max_{i\in\mathcal{M}} \left\{\mathcal{C}_i(t-1)\right\}$ if they are not targeted by $\kappa(t)$. The set can be mathematically defined as follows:
\begin{align}\label{eq:critical-set}
\mathcal{P}(t) = \Bigg\{i\in\mathcal{M} \;\Biggm|\; &\frac{|\mathcal{W}_i(0)| + \left(D_i(t-1)+1\right) - p_i}{1-p_i} \nonumber\\
 & \qquad >  \frac{|\mathcal{W}_j(0)| + D_j(t-1)-p_j}{1-p_j}\Bigg\}\;,
\end{align}
where $j =   \arg\max_{k \in \mathcal{M}} \left\{ \frac{|\mathcal{W}_k(0)| + D_k(t-1) - p_k}{1-p_k}\right\}$. We will refer to this set as the ``highly critical set''. Also, define $\tau(\kappa(t))$ as the set of users that are targeted by the transmission $\kappa(t)$. The following theorem defines a maximum weight clique algorithm that can satisfy the critical criterion.
\begin{theorem}\label{th:critical-criterion}
The critical criterion in \eref{eq:heuristic-criterion} can be achieved by selecting $\kappa^*(t)$ according to the following optimization problem:
\begin{align}\label{eq:criterion-optimization}
\kappa^{*}(t) &= \arg\max_{\kappa(t) \in \mathcal{G}(t)} \left\{\sum_{i \in \mathcal{P}(t) \cap \tau(\kappa(t))} \log\left(\cfrac{1}{p_i}\right)\right\}.
\end{align}
In other words, the transmission $\kappa(t)$ that can satisfy the critical criterion can be selected using a maximum weight clique problem in which the weight of each vertex $v_{ij}$ in $\mathcal{P}(t)$ can be expressed as:
\begin{align}\label{eq:weights}
w_{ij}^* = \log\left(\frac{1}{p_i}\right)= -\log(p_i).
\end{align}
\end{theorem}
\begin{proof}
Users $j \in \mathcal{M} \setminus \mathcal{P}(t)$ are unable to increase $\max_{i\in\mathcal{M}}\left\{\mathcal{C}_i(t)\right\}$ compared to $\max_{i\in\mathcal{M}}\left\{\mathcal{C}_i(t-1)\right\}$ with probability 1, even if they experience a decoding delay. This is true since the set $\mathcal{P}(t)$ is constructed such that it contains all users that have non-zero probabilities of increasing the completion time. According the definition of $\mathcal{C}_i(t)$ in \eref{eq:Ct}, $\forall~i\in\mathcal{M}$, users $i \in \mathcal{P}(t)$ will not increase $\max_{i\in\mathcal{M}}\left\{\mathcal{C}_i(t)\right\}$ after the transmission $\kappa(t)$ only if they do not experience a decoding delay increment in this transmission. Consequently, we get:
\begin{align}
& \mathds{P}\left(\max_{i\in\mathcal{M}}\left\{\mathcal{C}_i(t)\right\} = \max_{i\in\mathcal{M}}\left\{\mathcal{C}_i(t-1)\right\}\right) \nonumber \\
& \qquad = \mathds{P}\left(\max_{i\in\mathcal{P}(t)}\left\{\mathcal{C}_i(t)\right\} = \max_{i\in\mathcal{M}}\left\{\mathcal{C}_i(t-1)\right\}\right) \nonumber \\
& \qquad = \mathds{P} \left(D_i(t)-D_i(t-1)=0, \forall~i\in\mathcal{P}(t)\right) \nonumber \\
& \qquad = \prod_{i \in \mathcal{P}(t)} \mathds{P} \left(D_i(t)-D_i(t-1)=0\right).
\end{align}
According to the analysis done in \cite{vtc}, the critical criterion in \eref{eq:heuristic-criterion} can be achieved by selecting $\kappa^*(t)$ according to the following optimization problem:
\begin{align}
\kappa^{*}(t) &= \arg\max_{\kappa(t) \in \mathcal{G}(t)} \left\{\sum_{i \in \mathcal{P}(t) \cap \tau(\kappa(t))} \log\left(\cfrac{1}{p_i}\right)\right\}.
\end{align}
\end{proof}

\subsection{Proposed Heuristic Algorithm} 
Despite the importance of the satisfaction of the critical criterion in order to minimize the probability of increase of the maximum individual completion time, it may not fully exploit the power of IDNC. In other words, once a clique is chosen according to \eref{eq:criterion-optimization} from the users in the highly critical set $\mathcal{P}(t)$, there may exist vertices belonging to other users that can form an even bigger clique. Thus, adding this vertex to the clique and serving these users will benefit them without affecting the IDNC constraint for the users belonging to $\mathcal{P}(t)$.  

To schedule such vertices and their users, we will use the multi-layer graph selection introduced in \cite{vtc} with modified layers. Let $\mathcal{G}_1,\mathcal{G}_2,...\mathcal{G}_h$ (with $h \in \mathds{N}$) be the sets of vertices of $\mathcal{G}(t)$, such that $v_{ik} \in \mathcal{G}_n$ if the following conditions are true:
\begin{itemize}
\item $\mathcal{C}_i(t-1)+ \cfrac{n}{1-p_i} > \mathcal{C}_j(t-1)$.
\item $\mathcal{C}_i(t-1)+ \cfrac{n-1}{1-p_i} \leq \mathcal{C}_j(t-1)$.
\end{itemize}
where $j =   \underset{ i \in \mathcal{M}}{\text{argmax }} \left\{ \mathcal{C}_i(t-1)\right\}$. Consequently, the IDNC graph at time $t$ is partitioned into $h$ layers with descending order of criticality. By examining the above condition, the vertices of the users of $\mathcal{P}(t)$ are all in layer $\mathcal{G}_1$. Moreover, the $n$-th layer of the graph includes the vertices of the users who may eventually increase $\max_{i\in\mathcal{M}}\left\{\mathcal{C}_i(t+n)\right\}$ if they experience $n$ decoding delay increments in the subsequent $n$ transmissions. Consequently, a user with vertices belonging to $\mathcal{G}_i$ is more critical than another with vertices belonging to $\mathcal{G}_j$, $j>i$, as the former has a higher chance to increase the overall completion time.

In order to guarantee the satisfaction of the critical criterion, the algorithm \cite{vtc} first finds the maximum weight clique $\kappa^*$ in layer $\mathcal{G}_1$ as mandated by Theorem \ref{th:critical-criterion}. We then construct $\mathcal{G}_2(\kappa^*)$ including each vertex in $\mathcal{G}_2$ that is adjacent to all vertices in $\kappa^*$ (i.e. forms a bigger clique with $\kappa^*$). After assigning the same weights defined in \eref{eq:weights}, the maximal weight clique in $\mathcal{G}_2(\kappa^*)$ is found and added to $\kappa^*$. This process is repeated for each layer $\mathcal{G}_i, i\leq h$ of the graph to find the selected maximal weight clique $\kappa^*\in\mathcal{G}(t)$ to be transmitted at time $t$.

\section{Simulation Results} \label{sec:results}

\begin{figure}[t]
\centering
  \includegraphics[width=1\linewidth]{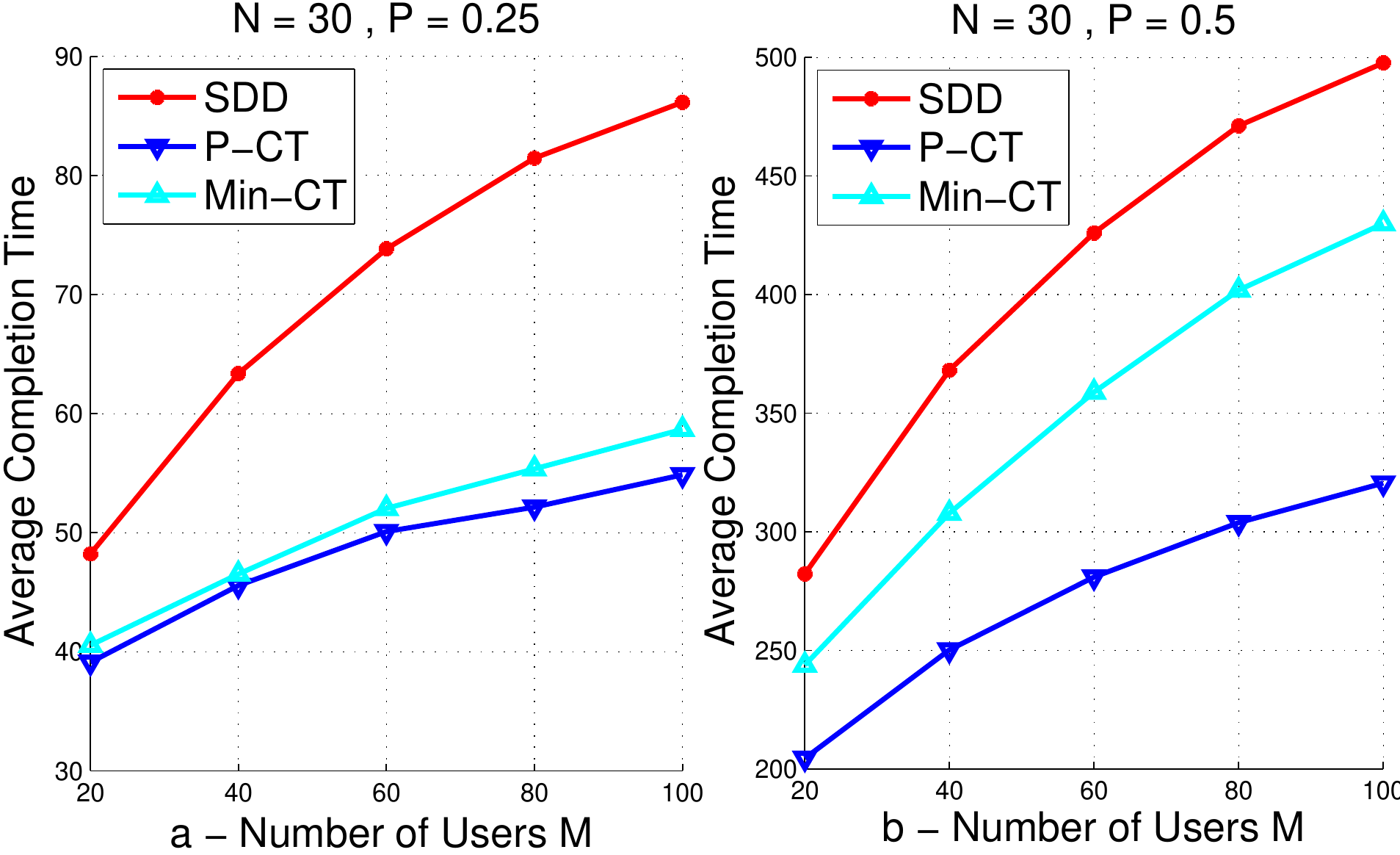}\\
  \caption{Mean completion time for IDNC versus number of users $M$.}\label{fig:M}
\end{figure}

\begin{figure}[t]
\centering
  \includegraphics[width=1\linewidth]{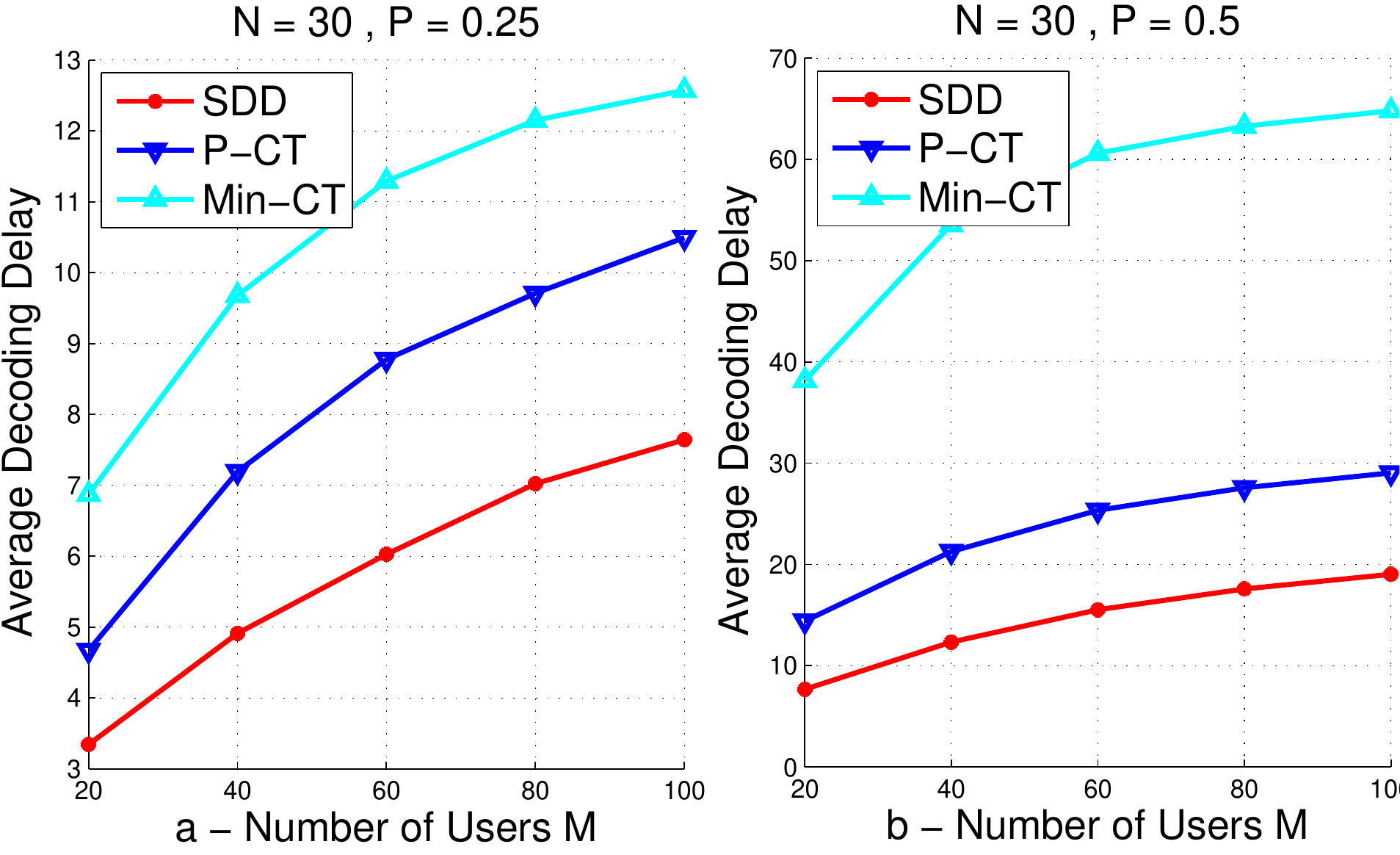}\\
  \caption{Mean decoding delay for IDNC versus number of users $M$.}\label{fig:M2}
\end{figure}

\begin{figure}[t]
\centering
  \includegraphics[width=1\linewidth]{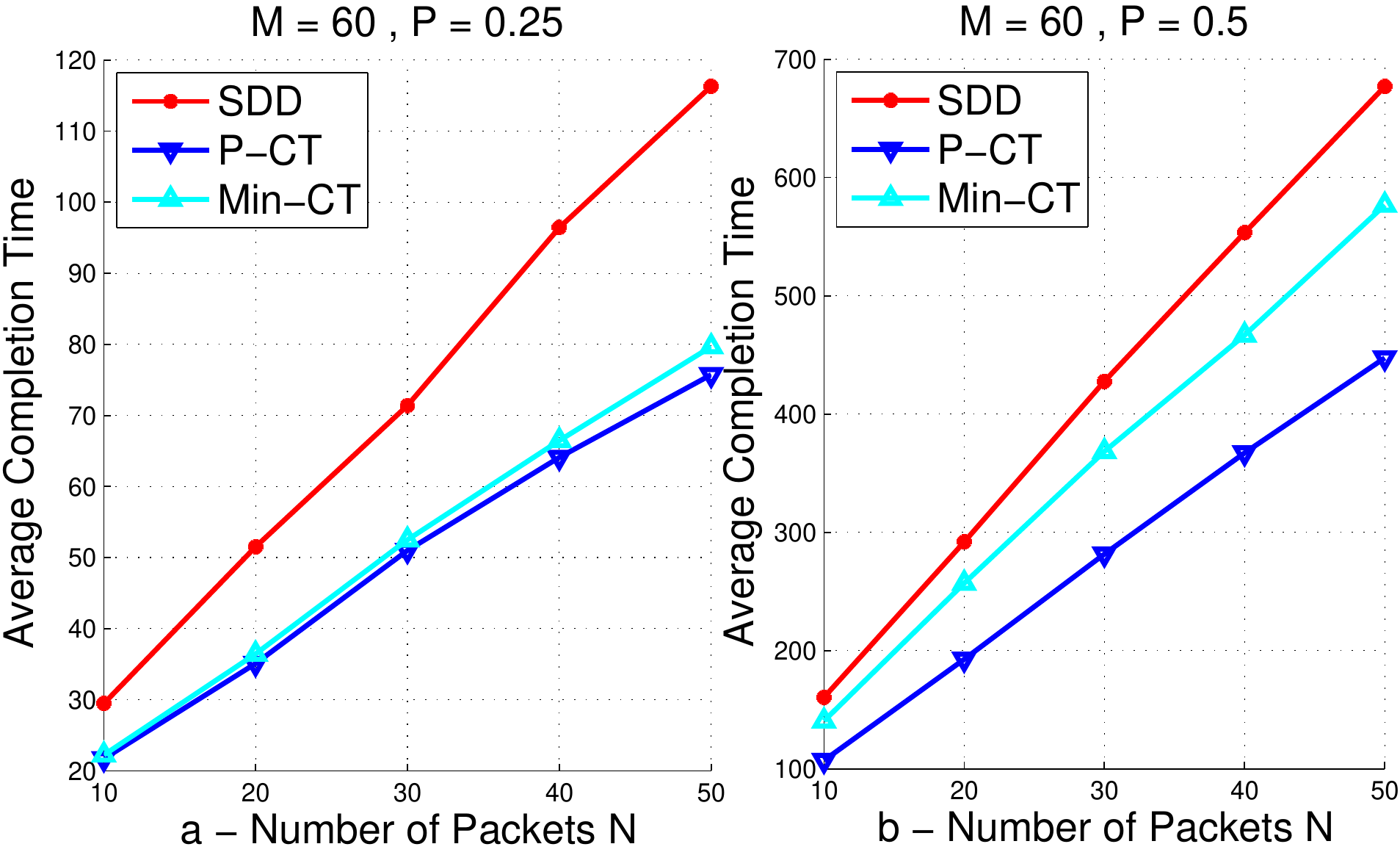}\\
  \caption{Mean completion time for IDNC versus number of packets $N$.}\label{fig:N}
\end{figure}

\begin{figure}[t]
\centering
  \includegraphics[width=1\linewidth]{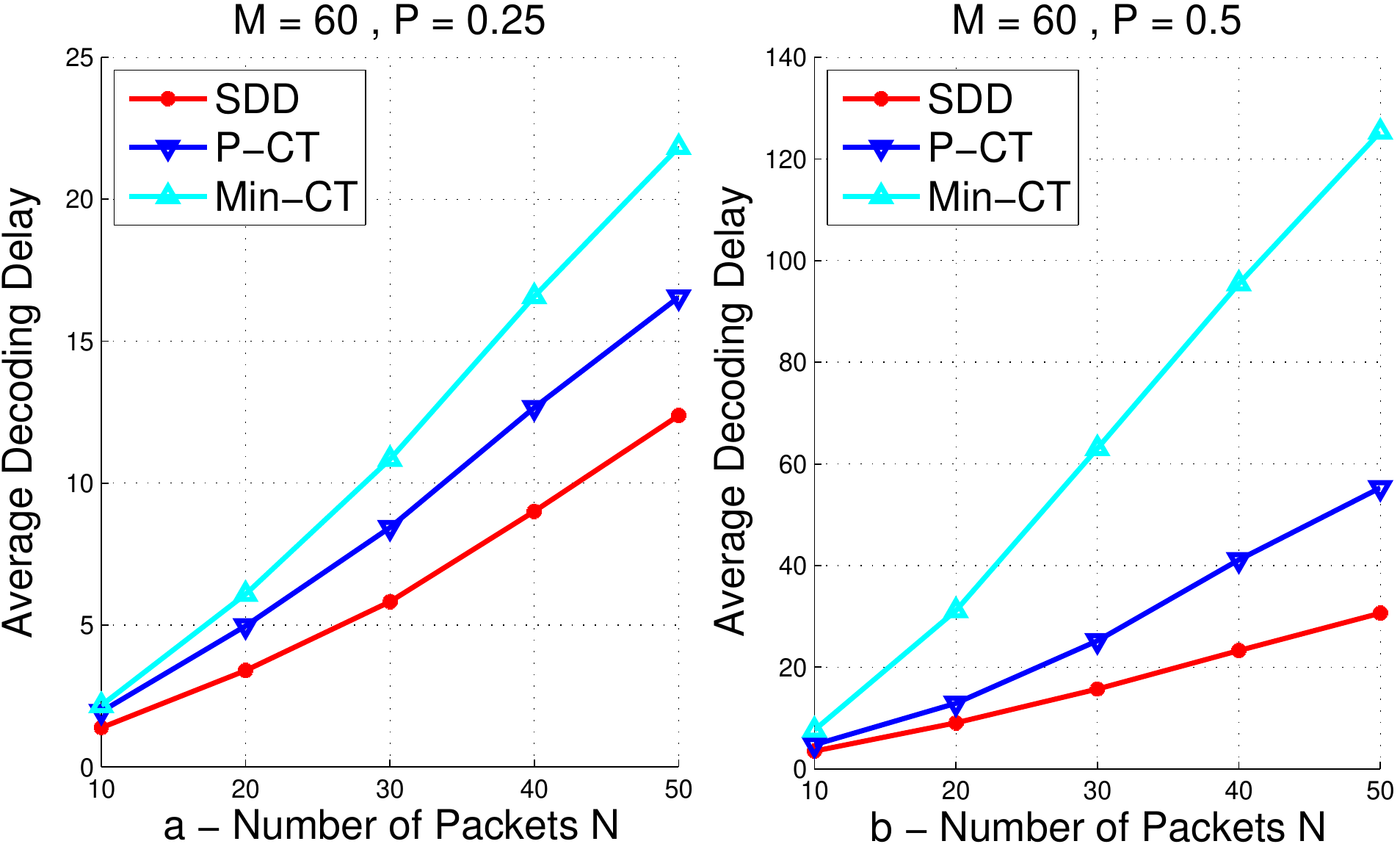}\\
  \caption{Mean decoding delay for IDNC versus number of packets $N$.}\label{fig:N2}
\end{figure}

\begin{figure}[t]
\centering
  \includegraphics[width=1\linewidth]{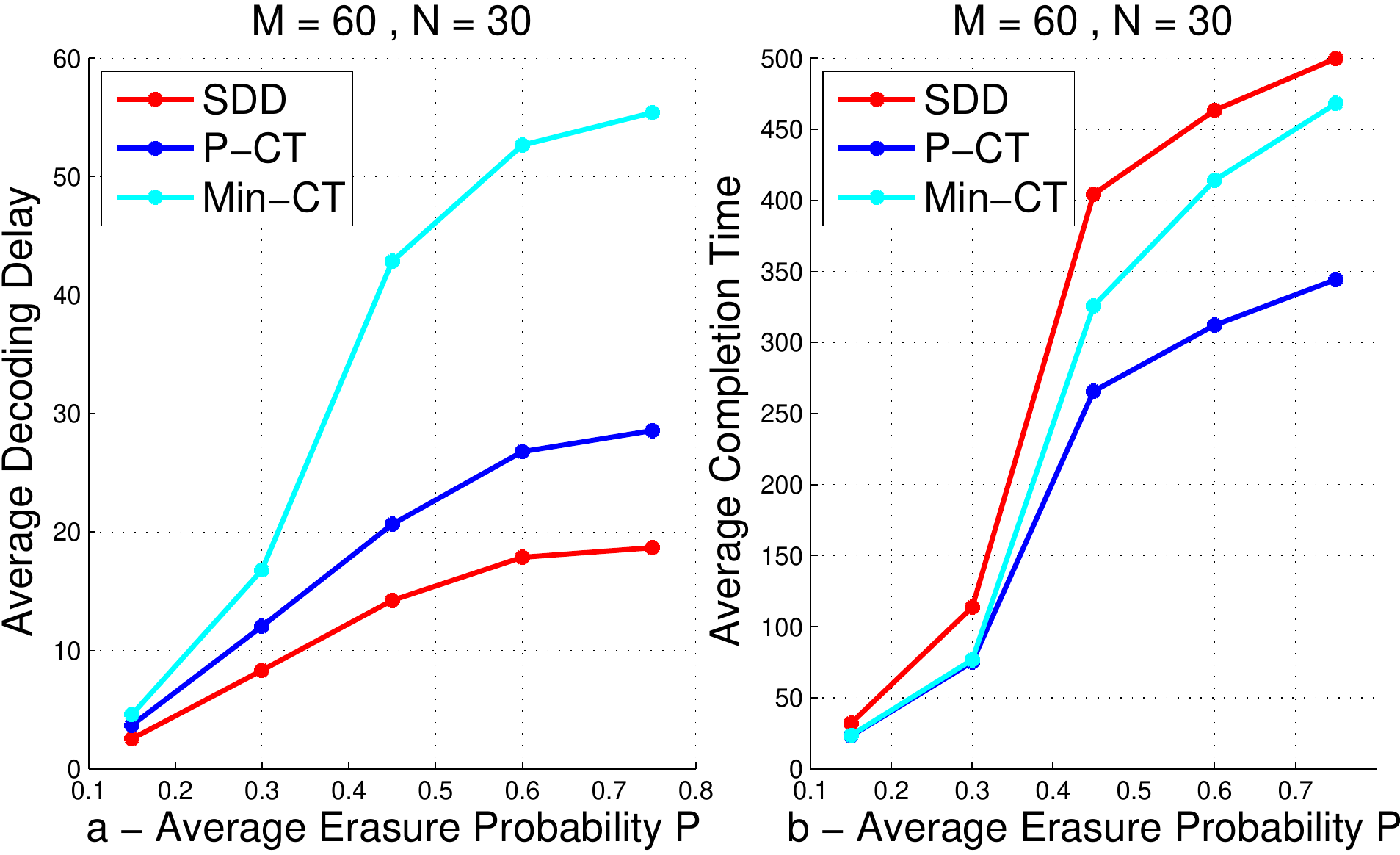}\\
  \caption{Mean delays for IDNC versus packet erasure probability $P$.}\label{fig:P}
\end{figure}

In this section, we present the simulation results comparing the different delays aspects achieved by the different policies to optimize each. We compare, through extensive simulations the sum decoding delay (denoted by SDD) and the completion time achieved by \cite{ref4} (denoted by Min-CT) and the completion time achieved by our algorithm (denoted by P-CT) while using the policy to reduce the sum decoding delay \cite{ref2} and the policy \cite{ref4} and our policy to reduce the completion time.

In all the simulations, the different delays are computed by frame then averaged over a large number of iterations. We assume that the packet erasure probability of all the users change from frame to frame while the average packet erasure probability $P$ remain constant.

\fref{fig:M} depicts the comparison of the mean completion time achieved by the policy to reduce the sum decoding delay (SDD), \cite{ref4} policy and our one to reduce the completion time (Min-CT and P-CT) against $M$ for $N=60$ and $P=0.25$ and $P=0.5$ receptively. \fref{fig:M2} illustrates the comparison of the decoding delay for the same inputs. \fref{fig:N} and \fref{fig:N2} depicts the comparison of the aforementioned delay aspects against $N$ for $M=60$ and $P=0.25$ and $P=0.5$ receptively and \fref{fig:P} illustrates this comparison against the erasure probability $P$ for $M=60$ and $N=30$.

From all the figures, we can clearly see that our proposed completion time algorithm outperforms the completion time policy proposed in \cite{ref4}. Moreover it gives the best agreement among the sum decoding delay and the completion in IDNC. 

\fref{fig:M}.a and \fref{fig:N}.a depicts the completion time when applying the sum decoding delay policy, the completion time policy \cite{ref4} and the our completion time policy against $M$ and $N$ for a low packet erasure probability. We see that the performance of P-CT and Min-CT are very close. Whereas in \fref{fig:M2} and \fref{fig:N2} where the sum decoding delay is computed for the same inputs, the performance of P-CT is much better than Min-CT one.

As the channel conditions become harsher (high packet erasure probability), our policy to reduce the completion time minimize the completion time better than the Min-CT. We can see from \fref{fig:M}.b, \fref{fig:M2}.b, \fref{fig:N}.b and \fref{fig:N2}.b that P-CT outperforms Min-CT in minimizing both the sum decoding delay and the completion time. \fref{fig:P}.a shows that for $P>0.3$, P-CT achieves a significant improvement in the completion time. This can be explained by the light of the P-CT policy characteristics. In the P-CT policy, the number of the erased packets is estimated using the law of large numbers. This approximation can be effective when the erasure of the channel or the input (number of packets and users) are  high enough.

\section{Conclusion} \label{sec:conclusion}

In this paper, we studied the effect of controlling the decoding delay to reduce the completion time below its currently best known solution. We first derived the decoding-delay-dependent completion time expressions. We then employed a heuristic that decides on coded packets by reducing the probability of decoding delay increase on a new layering of the IDNC graph based on user criticality in increasing the overall completion time. Simulation results showed that this new algorithm achieves a lower mean completion time and mean decoding delay compared to the best known completion time heuristics, with significant gains in harsh erasure scenarios.

\appendices

\section{Proof of Theorem 1}

Let us first define $\mathcal{E}_i(t)$ as the cumulative number of transmitted packets from the sender that were erased at user $i$ until time $t$. It is easy to infer that the reception completion event at time $t=C_i(S)$ of a user $i$ will occur when it receives an instantly decodable packet in the $C_i(S)$-th recovery transmission from the sender. Consequently, $\forall~t<=\mathcal{C}_i(S)-1$, the transmission at time $t$ following the schedule $S$ can be one of the following options:
\begin{itemize}
\item The packet can be erased at user $i$ $\Rightarrow$ The transmission will increase $\mathcal{E}_i(t)$ $\left(\mbox{i.e. } \mathcal{E}_i(t)=\mathcal{E}_i(t-1)+1\right)$.
\item The packet can be successfully received by the user $\Rightarrow$ Two cases can occur types:
\begin{itemize}
\item The packet is instantly decodable for user $i$. Note that user $i$ needs to receive $|\mathcal{W}(0)|-1$ of those packets until time $t=C_i(S)-1$ in order to complete its reception by the last missing source packet from the transmitted packet at time $t=C_i(S)$. Consequently, the number of such packets received by user $i$ until time $t=C_i(S)$ is equal to $|\mathcal{W}(0)|$.
\item The packet is either non-innovative or non instantly decodable $\Rightarrow$ This will increase the value of $D_i(S)$ by one each time it occurs until the reception completion for this user.
\end{itemize}
\end{itemize}

Consequently, the number of recovery transmission sent by the sender following schedule $S$ until user $i$ complete its reception of the frame packets (i.e. completion time of user $i$) can be expressed as follows:
\begin{align}\label{cti}
\mathcal{C}_i(S) = |\mathcal{W}_i(0)| + D_i(S) + \mathcal{E}_i(\mathcal{C}_i(S)-1)\;.
\end{align}

Let $\mathcal{X}_i(t)$ be a Bernoulli random variable that takes the value $1$ if the transmission at time $t$ is erased at user $i$. The definition of the variable is the following:
\begin{align}
\mathds{P}(\mathcal{X}_i(t) = x) =
\begin{cases}
p_i \hspace{0.9 cm}& \text{if } x = 1 \\
1-p_i \hspace{0.9 cm}& \text{if } x = 0 \\
\end{cases}
\end{align}

Consequently, the number of erased packets $\mathcal{E}_i(\mathcal{C}_i(S)-1)$ at user $i$ until $t=C_i(S)-1$ is therefore the sum of these $\mathcal{C}_i(S)-1$ Bernoulli trials. In other words,
\begin{equation}
\mathcal{E}_i(\mathcal{C}_i(S)-1)= \sum\limits_{t=1}^{\mathcal{C}_i(S)-1}\mathcal{X}_i(t)
\end{equation}
For large enough frame size $N$, the completion time $C_i(S)$ would also be large enough and thus $\mathcal{E}_i(\mathcal{C}_i(S)-1)$ can be approximated using the law of large numbers as follows:
\begin{align}
\mathcal{E}_i(\mathcal{C}_i(S)-1) \approx p_i(\mathcal{C}_i(S)-1).
\end{align}
Substituting the previous expression in \eref{cti} and re-arranging the terms, the completion time for user $i$ can be finally expressed as:
\begin{align}
\mathcal{C}_i(S) \approx \cfrac{|\mathcal{W}_i(0)| + D_i(S) - p_i}{1-p_i}.
\end{align}

Thus, the expression for the overall completion time can be expressed as:
\begin{align}
\mathcal{C}(S) \approx \max_{i\in\mathcal{M}}\left\{\frac{\left|\mathcal{W}_i(0)\right| + D_i(S)-p_i}{1-p_i}\right\}
\end{align}

\bibliographystyle{IEEEtran}
\bibliography{references}

\end{document}